\documentclass[11pt]{article}
\usepackage{amsmath,amsfonts,amsthm,filecontents}

\usepackage[a4paper]{geometry}

\usepackage{algorithm,tikz}
\usepackage[noend]{algorithmic}

\newtheorem{example}{Example}

\newtheorem{theorem}{Theorem}
\newtheorem{lemma}[theorem]{Lemma}
\newtheorem{corollary}[theorem]{Corollary}
\newtheorem{definition}[theorem]{Definition}
\newcommand{\R}{\mathbb R}
\newcommand{\N}{\mathbb N}

\def\calM{\mathcal{M}}

\DeclareMathOperator{\new}{new}
\DeclareMathOperator{\opt}{opt}

\begin{document}
\title{Adaptive Drift Analysis\thanks{This work was begun while both authors were visiting the ``Centre de Recerca Matem\'atica de Catalunya''. It profited greatly from this ideal environment for collaboration. A preliminary announcement of the result (without proofs) appeared in~\cite{PPSNthis}.
The work described in this paper was partly supported by EPSRC Research Grant
(refs\ EP/I011528/1) ``Computational Counting''
}}
\author{Benjamin Doerr\\
Max Planck Institute for Computer Science,\\
Campus E1 4\\
66123 Saarbr\"ucken, Germany
 \and 
Leslie Ann Goldberg\\
Department of Computer Science\\
University of Liverpool\\
Ashton Bldg, Liverpool L69 3BX, UK
}

\newcommand\eps{\varepsilon}
\maketitle

\begin{abstract}
We show that, for any $c>0$, the (1+1)~evolutionary algorithm using an arbitrary
mutation rate $p_n = c/n$  finds the optimum of a
linear objective function over bit~strings of length~$n$ in expected time $\Theta(n \log n)$. 
Previously, this was only known for $c\leq 1$.
Since previous work also
shows that universal
drift functions cannot exist for $c$ larger than a certain constant,
we instead define drift functions  which depend 
crucially on the relevant objective functions (and also on~$c$ itself). 
Using these carefully-constructed drift functions, we
prove that the expected optimisation time is $\Theta(n \log n)$.
By giving an alternative proof of the multiplicative drift theorem, we also show that our
optimisation-time bound holds with high probability.  \end{abstract}

\section{Introduction}

Drift analysis is central to the field of evolutionary algorithms.
This type of analysis was implicit in the work of Droste, Jansen and Wegener~\cite{DJW02},
who analysed 
the optimisation of linear functions over bit~strings by the classical (1+1) evolutionary algorithm ((1+1)~EA)
with mutation rate $p_n =1/n$.
The method was made explicit in the work of He and Yao who gave a simple, clean analysis.
Later fundamental applications of drift analysis in the theory of evolutionary computation include~\cite{GielL06,  GielWegenerSTACS03,  HappJKN08, NeumannOW09, OlivetoW11}.

Recent work by Johannsen, Winzen and the first author~\cite{DoerrJW10cec,DoerrJW10} shows that  drift analysis, as it is currently  used,    relies strongly on the 
fact that the mutation probabilities~$p_n$ are relatively small. 
As He and Yao observed~\cite{HeY02},
the  analysis in~\cite{HeY01}  only applies if the mutation 
probability $p_n$ is  strictly smaller than $1/n$, where $n$ is the length of the bit~strings of the search space. 

This  restriction was improved in~\cite{HeY04}, where a family of drift functions was presented 
that works for the most common  mutation probability $p_n = 1/n$.\footnote{Note, though, that in that paper an EA only accepting strict improvements was analysed; this fact was exploited in the proof. We have little doubt, though, that their proof can be adapted to work also for the more common setting that also an offspring with equal fitness is accepted.}
However,
as Doerr et al. have observed~\cite{DoerrJW10},   
this family of drift functions still ceases to work for $p_n \geq 4/n$.    Furthermore~\cite{DoerrJW10}, if 
$p_n > 4/n$, then for \emph{any} universal family of drift functions  (from the class of log-of-linear functions) 
there is a linear objective function~$f$,  and a search space element~$x$, 
such that the drift  from~$x$ is negative (so the proof that the (1+1) EA converges quickly does not go through).
Doerr et al. have also shown~\cite{DoerrJW10cec} that
this problem \emph{cannot} be fixed by applying the averaging approach of J\"agersk\"upper~\cite{Jagerskupper11} ---
that approach  fails for $p_n \ge 7/n$.
Thus, prior to the work presented here, it was an open problem whether the (1+1) EA 
minimises linear objective functions over bit~strings in $O(n \log n)$ time when the
mutation probability is $p_n = c/n$ for $c\geq 7$.

Our main result shows that this is the case.
Since it is known that no universal family of drift functions exists, we instead
manage to define a feasible family of drift functions in such a way that
the drift function~$\Phi_f$ depends crucially on the objective function~$f$.
Using this idea,  we show (see Theorem~\ref{thm:main}) that, for any constant~$c$,
    the (1+1) EA with  
 mutation probability $p_n=c/n$ optimises any
 family of linear objective functions over bit~strings in expected time $O(n \log n)$. 
  A corresponding lower bound follows easily from standard arguments, see  Theorem~\ref{thm:lower}.
  Thus, our result is as good as possible (up to a constant factor).
  
  By reproving a multiplicative drift theorem (which was first used to analyse evolutionary algorithms in~\cite{DoerrJW10}), we also  show that
 our bound on the optimisation time
 holds with high probability.
 The tail bounds in our drift theorem can also be used to show that many other  known bounds on optimisation 
 times also hold with high probability. This has been done for the (1+1) EA finding minimum spanning trees, computing shortest paths or Eulerian cycles in~\cite{PPSNtail}.

\section{Drift Analysis}

In this section, we give a brief description of drift analysis, which is sufficient for our purposes. For a more general
background to drift analysis, we refer to the papers cited above. 
 
\subsection{The (1+1) evolutionary algorithm}\label{subsecea}

Let~$F$ be a set of \emph{objective functions}. Each $f\in F$ is associated with a \emph{problem size} $n(f)\in \N$
and is a function from the \emph{search space} $\Omega_{f}$ 
to $ \R^{\geq 0}$.
Given $f$, 
the goal is to find an element $x\in \Omega_{ f}$ such that $f(x)$ is \emph{minimised}.
Our assumption that the optimisation problem is \emph{minimisation} (as opposed to maximisation) is without loss of generality, as is our assumption that the range of each objective function contains only non-negative numbers. For each objective function~$f$, let $\Omega_{\opt,f} \subseteq \Omega_{f}$ denote the set of \emph{optimal} search points --- that is, those that minimise the value of $f$.

\begin{definition}
\label{def:bitstrings}
We say that $F$ is a family of objective functions 
\emph{over bit~strings}
if, for every $ f\in F$,  
$\Omega_{f} = \{0,1\}^{n(f)}$. 
In this case, an element $x\in \Omega_f$
is a string of $n(f)$ bits, $x = x_{n(f)} \ldots x_1$. 
\end{definition}

\begin{definition}
\label{def:lin}
Suppose that $F$ is a family of objective functions over bit~strings.
We say that $F$ is \emph{linear}
if   each  $f\in F$  
is of the form $f(x) =  \sum_{i=1}^{n(f)} a_{i} x_i$,
where the coefficients $a_{i}$  are real numbers.
Without loss of generality, we assume that $a_{i+1}\geq a_{i}>0$ for all $i\in\{1,\ldots,n(f)-1\}$.
\end{definition}
 
\begin{example}
\label{example:binval}
Suppose, for $n\in \N$, that $f_n:\{0,1\}\rightarrow \R^{\geq 0}$ is defined
by $f_n(x_n \ldots x_1) = \sum_{i=1}^n 2^{i-1} x_i$. Then $F = \{f_n\}$ is a linear family of
objective functions over bit~strings. The value of $f_n(x)$ is the binary value of the
bit~string $x = x_n \ldots x_1$.
\end{example} 

\begin{example}
\label{example:onemax}
Suppose, for $n\in \N$, that $f_n:\{0,1\}\rightarrow \R^{\geq 0}$ is defined
by $f_n(x_n \ldots x_1) = \sum_{i=1}^n x_i$. Then $F = \{f_n\}$ is a linear family of
objective functions over bit~strings. The value of $f_n(x)$ is the number of ones in  the
bit~string $x = x_n \ldots x_1$.\end{example}
 
The randomised search heuristic that we study is the well-known (1+1) EA. 
To emphasize the role of the parameters, we refer to this algorithm as the \emph{(1+1) EA for minimising $F$}.
Given an objective function~$f\in F$,
this algorithm starts with an initial solution $x$, chosen uniformly at random from the search space $\Omega_{f}$. 
In each iteration, from its existing solution $x$, it generates a new solution $x'$ by \emph{mutation}. 

\begin{definition} Suppose $F$ is a family of objective functions over bit~strings
and that $p_n\in[0,1]$ for $n\in \N$.
In \emph{independent bit mutation}, each bit $x_i$ of~$x$ is flipped independently with probability $p_n$.
In other words, for each $i \in \{1,\ldots,n\}$ independently, we have $\Pr(x'_i = 1 - x_i) = p_n$ and $\Pr(x'_i = x_i) = 1-p_n$.
Often, $p_n=1/n$, but we do not make this assumption.
\end{definition}

In the subsequent \emph{selection} step, if $f(x') \le f(x)$, the EA \emph{accepts} the solution~$x'$, meaning that the  next iteration starts with 
$x_{\new} := x'$. Otherwise, the next iteration starts with $x_{\new} := x$. Since we are interested in determining  
the number of iterations that are necessary to find an optimal solution, we do  not specify a termination criterion here. A pseudo-code description of the (1+1)~EA is given in Algorithm \ref{alg:oneoneea}

\renewcommand{\algorithmicloop}{\textbf{repeat forever}}
\begin{algorithm}[h]
    \caption{The (1+1) EA for minimising $F$ over  bit~strings with independent bit mutation}
    \algsetup{indent=1.5em}
    \begin{algorithmic}[1]
    \STATE Input an objective function $f\in F$ .
        \STATE \textbf{Initialization:} Choose $x \in \{0, 1\}^{n(f)}$ uniformly at random.
        \LOOP
            \STATE Create $x' \in \{0, 1\}^{n(f)}$ by copying $x$.
            \STATE \textbf{Mutation:} Flip each bit in~$x'$ independently with probability $p_{n(f)}$.
            \STATE \textbf{Selection:} \textbf{if} $f(x') \le f(x)$ \textbf{then} $x := x'$.
        \ENDLOOP
    \end{algorithmic}
    \label{alg:oneoneea}
\end{algorithm}

 Note that the (1+1) EA is not typically used to solve
 difficult optimisation problems in practice. 
 There are other, more complex, search heuristics which are better for
 such problems in practice.  However, understanding the optimisation behaviour 
 of the (1+1) EA 
 often helps us to predict the optimisation behaviour   of more complicated EAs 
 (which are mostly too complex to allow rigorous theoretical analysis). As such, the (1+1) EA proved to be an important tool that attracted significant research efforts (see, e.g.,~\cite{Baeck93, DJW98, DJW02} for some early works).
  
\subsection{A simple drift theorem with tail bounds}
\label{sec:driftdef}

The \emph{optimisation time} of the (1+1) EA for minimising $F$
is  defined to be the   
number of times that  the objective function is evaluated before
the optimum is found.
This is (apart from an additive deviation of one) equal to the number of mutation-selection iterations.
Suppose that $c$ is a positive constant and that $F$ is a family of linear objective functions over bit strings.
Our main result (Theorem~\ref{thm:main})  shows that the (1+1) EA for minimising $F$ with 
independent bit-mutation rate $p_n = c/n$ 
has expected optimisation time $O(n(f) \log n(f))$.
It also shows that, with high probability, the optimisation time is of this order of magnitude.

In order to prove the main result,  we introduce the notion of \emph{piece-wise polynomial drift}. This will
be explained in Section~\ref{subsecpiecedrift}. In this section, we prepare the groundwork, by introducing the basic
drift theorems that we will need. We start by defining the notion of a \emph{feasible} family of drift functions.
When feasible families of drift functions exist, they allow an elegant analysis yielding upper bounds for the 
optimisation time of EAs.

\begin{definition}
\label{def:feasdrift}
 Let $\nu : \N \to \R^{\geq 0}$ be monotonically increasing and consider a family  $F$ of objective functions.
  For each $f\in F$, let
  $\Phi_f$ be a function from~$\Omega_{f}$ to~$\R^{\geq 0}$.
  We say that $\Phi=\{\Phi_f\}$ is a   
  $\nu$-feasible family of drift functions for
  a (1+1) EA for minimising $F$,
  if  there is an $n_0\in \N$ such that, for every $f\in F$ with $n(f)\geq n_0$, the
  following conditions are satisfied.
       \begin{enumerate}
        \item $\Phi_f(x) = 0$ for all $x \in \Omega_{\opt,f}$;
        \item $\Phi_f(x) \ge 1$ for all $x \in \Omega_{f} \setminus \Omega_{\opt,f}$;
        \item 
   for all $x \in \Omega_{f} \setminus \Omega_{\opt,f}$, \[E[\Phi_f(x_{\new})] \le \left(1 - \frac {1}{\nu(n(f))}\right) 
   \Phi_f(x),\] where, as above, we denote by $x_{\new}$ the solution resulting from executing a single iteration (consisting of mutation and selection) with initial solution~$x$.
\end{enumerate}
\end{definition}
 
Here is a simple example.
\begin{example} 
\label{example:simple}
Fix a positive constant~$c$.
Let $F$ be a linear family of objective functions over bit~strings
and consider the (1+1) EA for minimising $F$ which uses independent bit mutation with 
$p_n  = c/n$. Suppose that, for each $f\in F$,  
the coefficient $a_{1}$ is at least~$1$.
Then the trivial family $\Phi$ with $\Phi_f = f$  is an $(n/c')$-feasible family of drift functions for  this EA, where $c' := c (1- (c/n))^{n-1} \approx c e^{-c}$. 
However, as we shall see, this not often a very useful family of drift functions.
\end{example}

The following well-known theorem  (Theorem~\ref{thm:classical}, below) 
shows how the optimisation time can be bounded using a drift function. 
Similar arguments appear in the context of coupling proofs. See, for
example, \cite[Section~5]{DyerGreenhill}.
Much more is known about drift analysis. See, for
example~\cite{Hajek82}. 
Note that Theorem~\ref{thm:classical} gives a probability tail bound
in addition to an upper bound on the expected optimisation time.
The tail bound is not new, but it seems to be unknown in the
evolutionary
algorithms literature. It can be applied to improve several previous
results (see~\cite{PPSNtail}).

\begin{theorem}
\label{thm:classical} 
Consider a family $F$ of objective functions
and a $\nu$-feasible family $\Phi$ of drift functions for
a (1+1) EA for minimising $F$.
Let $\Phi_{\max,f}$ denote $\max\{\Phi_f(x) \mid x \in \Omega_{f}\}$.  
Then there is an $n_1\in \N$ such that, for every $f\in F$ with $n(f)\geq n_1$,
the expected optimisation time 
of the EA
is at most 
$$ {\nu(n(f))}  ( \ln \Phi_{\max,f}+1).$$ 
Also, for any $\lambda >0$,  the probability that the optimisation time exceeds 
$$\lceil {\nu(n(f))}  (\ln \Phi_{\max,f} + \lambda  ) \rceil$$  is  at most $ \exp(-\lambda ))$.
\end{theorem}

\begin{proof} 

Let $n_0$ be the value from  Definition~\ref{def:feasdrift}.
Definition~\ref{def:feasdrift} rules out the possibility 
that $\max\{\nu(n) \mid n\geq n_0\}<1$.
Also, if $\max\{\nu(n) \mid n\geq n_0\}=1$ then, from part (3) of
the definition, $E[\Phi_f(x_{\new})]=0$ so the optimisation time is~$1$.
Suppose then, that there is an $n\in\N$ such that $\nu(n)>1$.
Let $n'_0$ be $\min\{n \in \N \mid \nu(n)>1\}$ (actually, it would suffice to
take $n'_0$ to be any member of this set, but, for concreteness, we take the minimum).
Let $n_1=\max(n_0,n'_0)$. Now consider any $f\in F$ with $n(f)\geq n_1$
and note that the first two conditions in  Definition~\ref{def:feasdrift} are satisfied.

Let $n=n(f)$.
Fix an arbitrary initial solution $x_0 \in \Omega_f$.
Consider  starting the EA with this initial solution~$x_0$ instead of choosing a random one.
Denote by $\Phi_{[t]}$ the value of $\Phi_f(x)$ after $t$ selection-mutation steps. 
Denote by $T_{\opt, x_0}$ the first time when the current solution $x$ is optimal.
Thus, from Definition~\ref{def:feasdrift}, $\Phi_{[T_{\opt,x_0}]}=0$, and
for $t<T_{\opt,x_0}$, we have $\Phi_{[t]}\geq 1$.
From the third condition in Definition~\ref{def:feasdrift},
$$E[\Phi_{[t]}] \leq {(1-1/\nu(n))}^t \Phi_{[0]} \leq {
(1-1/\nu(n))}^t 
\Phi_{\max,f} \le \exp(-t/\nu(n)) \Phi_{\max,f},$$ 
where, in the last estimate, we used the well-known inequality $1+z \le e^z$, which is valid for all $z \in \R$.

It is well known (see, for example~\cite[Problem 13(a), Section 3.11]{GS}) that
if $X$ is a random variable taking values in the non-negative integers, then $E[X] = \sum_{i = 1}^\infty \Pr(X \ge i)$.
Therefore,
the expected optimisation time $E[T_{\opt, x_0}]$ can be written as
$$E[T_{\opt,x_0}] = \sum_{i\geq 1} \Pr(T_{\opt,x_0}\geq i) = \sum_{t \ge 0} \Pr(\Phi_{[t]}>0).$$  
So, for any  non-negative integer~$T$,
$E[T_{\opt,x_0}] \leq  T+\sum_{t\ge T}\Pr(\Phi_{[t]}>0)$.
Since, by Markov's inequality, 
$\Pr(\Phi_{[t]} > 0) = \Pr(\Phi_{[t]} \ge 1) \le E[\Phi_{[t]}]$,  
$$E[T_{\opt,x_0}] \leq    T+\sum_{t\ge T} E[\Phi_{[t]}].$$

Now let $T=\lceil \ln(\Phi_{\max,f}) \nu(n)  \rceil = \ln(\Phi_{\max,f}) \nu(n)  +\eps$ for some $0 \le \eps < 1$. By our  upper bounds above, we obtain
$$ E[T_{\opt, x_0}] \le  T + (1- 1/\nu(n))^{T} \Phi_{\max,f} \sum\nolimits_{i = 0}^\infty (1 -  1/\nu(n))^i .$$
Since $\nu(n)>1$,  $\sum_{i = 0}^\infty (1 -  1/\nu(n))^i  = \nu(n)$.
Plugging this in with the definition of $T$ and using  
$(1- 1/\nu(n))^{\ln(\Phi_{\max,f}) \nu(n)} \leq
\exp^{- \ln(\Phi_{\max,f})} = 1/\Phi_{\max,f}$,
\begin{align*} E[T_{\opt, x_0}] &\le
\ln(\Phi_{\max,f}) \nu(n)  + \eps   + (1 -  1/\nu(n))^\eps  \nu(n)\\
&=
\nu(n) \left(
\ln(\Phi_{\max,f})   + \eps/\nu(n)   + (1 -  1/\nu(n))^\eps  \right).
\end{align*}
We can now check, for every $\eps \in [0,1]$,
that $\eps/\nu(n)   + (1 -  1/\nu(n))^\eps \leq 1$, as required.
This is easiest seen by checking it for $\eps = 0$ and $\eps = 1$ and noting that the term is convex in~$\eps$.
Finally,  let $T' := \lceil (\nu(n) )(\ln(\Phi_{\max,f}) + \lambda   )\rceil$ for $\lambda > 0$. We compute
\begin{align*}
        \Pr(T_{{\opt}, x_0} > T') & = \Pr(\Phi_{[T']} > 0) \le E[\Phi_{[T']}] \le \exp(-T'  /\nu(n)) \Phi_{\max,f} \le \exp(-\lambda ).
\end{align*} 
\end{proof}

The proof above uses the argument $E[\Phi_{[t]}] \le (1 - 1/\nu(n))^t \Phi_{\max,f}$. This had been used previously in the so-called \emph{methods of expected weight decrease}~\cite{NeumannW07}. There, however, it was followed up with a simple Markov inequality argument that led to a bound on the expected run-time that is weaker (by a constant factor) than what our drift theorem yields. Hence the main difference between the two approaches is that ours gives a better transformation of the drift of $E[\Phi_{[t]}]$ into a bound on $E[\min\{t \mid \Phi_{[t]} < 1\}]$. Note, just to avoid misunderstandings, that typically 
$E[\min\{t \mid \Phi_{[t]} < 1\}]$ and $\min\{t \mid E[\Phi_{[t]}] < 1\}$ are different quantities. 

Theorem~\ref{thm:classical}  indicates that a 
family of drift function is better if 
the maximum values  $\Phi_{\max,f}$ are small. 
In  Example~\ref{example:simple},  taking $\Phi_f = f$ only yields an 
upper bound $O(n(f) \log f_{\max})$ for the expected optimisation time,
where $f_{\max} = \max\{ f(x) \mid x \in \Omega_{f}\}$.
This can be a weak bound. 
For example,
applying it to the family~$F$ from Example~\ref{example:binval} yields
a bound
$O(n(f)^2)$ for the expected optimisation time (which, as we shall see,    is a weak bound).

\subsection{Drift analysis for linear objective functions over bit~strings}
\label{sect:return}

The main goal of this paper is to analyse the optimisation time of the (1+1) EA for minimising a linear family $F$
of objective functions over bit~strings, assuming independent-bit mutation with $p_n=c/n$ (for a fixed constant~$c$).
The reason for assuming $p_n=c/n$ is that results of Droste, Jansen and Wegener (Theorem 13 and 14 in~\cite{DJW02}) 
show that this is the optimal order of magnitude.
Since our objective is an $O(n(f) \log n(f))$ bound on optimisation time,
we ease the language with the following definition.

\begin{definition}
A \emph{feasible family of drift functions} is
a family of drift functions which is $\nu$-feasible
for a function $\nu(n) = O(n)$.
\end{definition}

Finding feasible drift functions is typically quite tricky.  
Doerr, Johannsen and Winzen 
built on earlier ideas of Droste, Jansen and Wegener~\cite{DJW02}
and He and Yao~\cite{HeY04} in order to show that,  
for any linear family $F$ of objective functions over bit~strings, 
the family $\Phi$ defined by
$$ \Phi_f(x)  = \sum_{i = 1}^{\lfloor n(f)/2\rfloor} x_i + \tfrac 54 \sum_{i = \lfloor n(f)/2 \rfloor+1}^{n(f)} x_i$$
is a feasible family of drift functions for
the (1+1) EA for minimising $F$ which uses independent bit mutation with $p_n=1/n$.
(Thus, this suffices for the case $c=1$.)

This family
$\Phi = \{\Phi_f\}$ is said to be a \emph{universal} family of feasible 
drift functions
because $\Phi_f$ depends on $n(f)$, but not otherwise on $f$.
Since $\Phi_{\max,f} = \Theta(n(f))$,  this gives an
expected optimisation time of $O(n(f) \log n(f))$, which is asymptotically optimal~\cite{DJW02}. 
Proving that this $\Phi$ is a feasible  family, while not trivial, is not overly complicated. 
This discovery of a universal family of feasible drift functions  gives an elegant analysis of the EA.

Unfortunately, even if we allow $\Phi_{\max,f}$ to grow faster than~$\Theta(n(f))$,
such universal families of feasible drift functions  only exist 
when $c$ is small (as noted in the introduction to this paper). 
For larger values of $c$, the function $\Phi_f$ has to depend upon~$f$.
Prior to this paper, no non-trivial drift functions of this form were known, so it  was an open problem whether
the $O(n(f) \log n(f))$ time bound also applies for $c>1$. We show that this is the case.

\subsection{Our result}

Our main theorem is as follows.

\begin{theorem}
\label{thm:main}
Let~$c$ be a positive constant.
Let~$F$ be a family of linear objective functions over bit strings.
The (1+1) EA for minimising $F$ with independent bit-mutation rate
$p_n = c/n$ 
has expected optimisation time $O(n(f) \log n(f))$.
There is a constant~$k$ and a function $\nu(n) =O(n)$
such that, for any $\lambda>0$,  the probability that the optimisation time
exceeds this bound by $k \nu(n) \lambda$ time steps is at most $k\exp(-\lambda)$.
\end{theorem}

We prove Theorem~\ref{thm:main} by constructing
a feasible family of drift functions 
for the EA
that is \emph{piece-wise polynomial} (a notion
that will be defined in Section~\ref{subsecpiecedrift}).
Lemma~\ref{lem:piecewise} extends Theorem~\ref{thm:classical} to piece-wise polynomial feasible families 
of drift functions, allowing us to prove Theorem~\ref{thm:main}.

  Theorem~\ref{thm:main} is interesting for two reasons. On the methodological side, 
  the proof of the theorem greatly enlarges our understanding about
 how to choose good
 drift functions. This might enable better solutions for some problems where drift analysis has not yet been very successful. 
 Examples are the minimum spanning tree problem~\cite{NeumannW07} and the single-criteria formulation of the single-source shortest path problem~\cite{BaswanaBDFKN2009}. For both problems, the known bounds on
 the expected optimisation time contain a $\log(f_{\max})$-factor, stemming from the fact that, at least implicitly, 
 drift analysis with the trivial family of drift functions with $\Phi_f=f$  is conducted.

Of course, our result is also interesting because it for the first time shows that 
 linear functions are optimised by the (1+1) EA in time $O(n(f) \log n(f))$, regardless of what mutation probability $p_n = c/n$ is used. Note that this is not obvious. In~\cite{DoerrPPSN10}, the authors show that already for monotone functions, a constant factor change in the mutation probability can change the optimisation time from polynomial to exponential.

\subsection{Piece-wise polynomial drift}\label{subsecpiecedrift}

Let $F$ be a family of linear objective functions over bit strings.
Let $\Phi$ be a feasible family of drift functions for a (1+1)-EA for minimising~$F$.

We start with an elementary observation about  $\Phi$,  which is that, in order  
to obtain an $O(n(f) \log n(f))$ bound on the expected optimisation time,
we do not really need  $\Phi_{\max,f}$ to be 
bounded from above by a polynomial in~$n(f)$ ---
we can afford to have a constant number of ``huge jumps''. The following arguments can be seen as a variation of the fitness level method~\cite{Wegener02}.

\begin{definition}
Fix $k\in \N$. Suppose that, for every $ f\in F$, 
$\mathcal{M}^f = M^f_0 , \ldots, M^f_k$
is a partition of   $\Omega_{f} $.
Let $\mathcal{M} = \{\calM^f \mid f\in F\}$.
We say~$\calM$  is a \emph{family of fitness-based $k$-partitions for~$F$} if for all $f \in F$,
\begin{enumerate}
 \item $M^f_0 = \{{\bf 0}\}$,
 \item for all $i<j$, $x \in M^f_i$ and $y \in M^f_j$, we have $f(x) < f(y)$.
\end{enumerate} 
\end{definition} 
 
We use the  notation $\min \Phi_f(M^{f}_j)$ to denote
$\min\{\Phi_f(x) \mid x \in M^{f}_j\}$ and the notation $\max \Phi_f(M^{f}_j)$
to denote $ \max\{\Phi_f(x) \mid x \in M^{f}_j\}$.

\begin{lemma}
Let $F$ be a family of linear objective functions over bit strings.
Let $\Phi$ be a $\nu$-feasible family of drift functions for a (1+1)-EA for minimising~$F$.
Let $\calM$ be a family of fitness-based $k$-partitions for~$F$.
Then there is an $n_1\in \N$ such that, for every $f\in F$ with $n(f)\geq n_1$, 
the expected optimisation time of the EA is at most
$$\nu(n(f)) \sum_{j = 1}^k \left(\ln(\max \Phi_f(M^{f}_j)) - \ln(\min \Phi_f(M^{f}_j))+1\right) .$$
Also, for any 
$\lambda >0$,  the probability that the optimisation time exceeds 
$$\  \sum_{j = 1}^k \left\lceil \nu(n(f)) \left(\ln(\max \Phi_f(M^{f}_j)) - \ln(\min \Phi_f(M^{f}_j) + \lambda \right)\right\rceil$$
 is  at most $k \exp(-\lambda )$.
 \label{lem:piecewise}
\end{lemma}

\begin{proof}
  
Let $n_1$ be the quantity in Theorem~\ref{thm:classical} 
(which is at least as large as the quantity $n_0$ in Definition~\ref{def:feasdrift}).
Let $f\in F$ with $n(f)\geq n_1$.  
 For $0\leq j \leq k$, let $\Omega_{f,j} = \bigcup_{\ell=0}^j M^f_\ell$ and let $\mu_{f,j} = \min \Phi_f(M^f_j)$.
For $1\leq j \leq k$,
define $\Psi_{f,j}  :\Omega_{f,j} \to \R$ as follows. 
If $\Phi_f(x)\geq \mu_{f,j}$ then $\Psi_{f,j}(x) = \Phi_f(x)/\mu_{f,j}$.
Otherwise, $\Psi_{f,j}(x)=0$.

Now for $j\in\{1,\ldots,k\}$, consider 
restricting the search space to $\Omega_{f,j}$.
Since the partition $\calM^f$ is fitness based, we conclude that,
if the EA is started with input~$f$, and an initial solution in $\Omega_{f,j}$,
all new solutions that are accepted by the EA are in $\Omega_{f,j}$.

Considering all solutions in $\Omega_{f,j-1}$ to be equivalent to the all-zero state ${\bf 0}$,
we note that $\{\Psi_{f,j}\mid  f\in F\}$ satisfies the first two conditions of being a $\nu$-feasible 
family of drift functions for~$F$ on~$\{\Omega_{f,j}\}$.
Also, if $\Phi_f(x)\geq \mu_{f,j}$
then $E[\Phi_f(x_{\new})] \le (1 - 1/\nu(n(f))) \Phi_f(x)$
so 
$$E[\Psi_{f,j}(x_{\new})] \leq
E[\Phi_f(x_{\new})/\mu_{f,j}] \leq (1-1/\nu(n(f))) \Psi_{f,j}(x).$$
So, by Theorem~\ref{thm:classical},
the expected time until a solution in $\Omega_{f,j-1}$ is reached is at most
$$\nu(n(f))(1+ \ln \max\{\Psi_{f,j}(x) \mid x\in \Omega_{f,j}\}),$$
which is at most
$$\nu(n(f)) \left(1+\ln \left(
\frac{\max \Phi_f(M^f_j)}
{\min \Phi_f(M^f_j)} \right)\right).$$

This gives the desired result, summing from $j=k$ down to $j=1$.

For the high probability statement, again from
Theorem~\ref{thm:classical}, we conclude that with probability at
least $1- \exp(-\lambda )$ , 
$$ \left\lceil {\nu(n(f))}\left(\ln\left(
\frac{\max \Phi_f(M^f_j)}{\min \Phi_f(M^f_j)}
\right) +  \lambda \right)\right\rceil$$
iterations suffice to go from a solution in $\Omega_{f,j}$ to one in $\Omega_{f,j-1}$. 

\end{proof}

\begin{definition}
\label{usedef}
Suppose that $\Phi$ is a family of feasible drift functions for~$F$.
We will say that $\Phi$ is \emph{piece-wise polynomial} (with respect to  the (1+1)-EA), 
if 
there is a constant~$k$  
and a family $\calM$ of fitness based $k$-partitions 
for $F$
such that for every $j\in\{1,\ldots,k\}$, 
$\ln(\max \Phi_f(M^f_j)) - \ln(\min \Phi_f(M^f_j)) = O(\log n(f))$.  
\end{definition}

If  $\Phi$ is a family of feasible drift functions for
a (1+1)-EA for minimising~$F$, and $\Phi$
is piece-wise polynomial with respect to the EA,
then the optimisation time bound given by Lemma~\ref{lem:piecewise} is $O(n(f) \log n(f))$.

\section{Construction of the Drift Function}

Let $F$ be a linear family of objective functions over bit~strings (see Definition~\ref{def:lin}). 
Fix a constant~$c$ and consider the (1+1) EA for minimising~$F$ with independent bit-mutation rate $p_n = c/n$.
We aim to construct a family~$\Phi$ of feasible drift functions
for the EA which is piece-wise polynomial with respect to  the EA.

\subsection{Notation and parameters}
\label{sec:param}

Recall from Definition~\ref{def:lin} that $\Omega_f = \{0,1\}^{n(f)}$ and that an element $x\in \Omega_f$
is written as a string of $n(f)$ bits, $x = x_{n(f)} \ldots x_1$.
In the proof, we shall often use the 
word ``left'' to refer to the most-significant bit (with the largest index, index~$n(f)$) of~$x$ and ``right'' to refer to the
least-significant bit (with the smallest index, index~$1$).

The proof will use several parameters, which we discuss here.
We start by fixing an arbitrarily-small positive constant~$\eps$. This is constant will be used
to precisely formulate the intermediate results.
To define the family~$\Phi$, we will use a sufficiently large constant~$K\geq 1$ (depending on~$c$ and~$\eps$)
and a sufficiently small positive constant~$\gamma$ (depending on~$c$, $\eps$ and $K$).

\subsection{Splitting into blocks}

The difficulty in defining a suitable drift function~$\Phi_f$
is that the optimisation of $f$ via the EA heavily depends on the coefficients $a_i$. 
If these are steeply increasing, as in Example~\ref{example:binval},
whether a new solution is accepted or not is determined by
the value of the leftmost bit that is flipped. On the other hand,
if these are of comparable size, as in Example~\ref{example:onemax},
the difference between the number of ``good'' bit-flips (turing a~$1$ into a~$0$) and
the number of ``bad'' bit-flips (turing a~$0$ into a~$1$) determines whether a new
solution is accepted.
Of course, the precise definitions of ``steeply increasing'' and ``comparable size'' depend on the 
constant~$c$ in the mutation probability. Also, an objective function~$f$ can be of a mixed type, having regions with steeply increasing coefficients and also regions where coefficients are of comparable size. 

Fix an objective function~$f$ with $n(f)=n$. 
To analyse~$f$ and define the corresponding
drift function~$\Phi_f$, we split the bit positions $\{1, \ldots, n\}$ into \emph{blocks}. The idea is that, within a block, one of the two behaviours is dominant. The definition of blocks, naturally, has to 
allow us to analyse the interaction between different blocks.

We first split the bit positions $\{1,\ldots,n \}$ into  \emph{miniblocks}. 
Start with $j=1$.
A miniblock starting at bit position~$j$ is constructed as follows.
If $a_{n}/a_j<n^2$, then $\{j,\ldots,n\}$ is a single miniblock.
Otherwise, let $i$ be the minimum value in $\{j+1,\ldots,n\}$ 
such that $a_i/a_j\geq n^2$. 
Then the set $\{j,\ldots,i\}$ is a miniblock. 
If $i=n$, we are finished. Otherwise,
set $j=i$ and repeat to form the next miniblock, starting at bit position~$j$.
Note that consecutive miniblocks overlap by one bit position.

The next thing that we do is merge consecutive pairs of miniblocks into 
\emph{blocks}.
To start out with, we just go through the miniblocks from right to left, 
making a block out of each pair of miniblocks. Note that this is (intentionally) different from just defining blocks analogous to miniblocks with the $n^{2}$ replaced by $n^{4}$. Note further that again consecutive blocks overlap in one bit position.

A block is said to be \emph{long} if it contains at least $\gamma n$ bit positions (recall that the
parameter~$\gamma$ is from Section~\ref{sec:param})
and \emph{short} otherwise.
It helps our analysis if any pair of long blocks has
at least three
short blocks in between.
So if two long blocks are separated by at most two short blocks, then we combine
the whole thing into a single long block. We repeat this (at most a constant number of times since there are less than $1/\gamma$ long blocks initially) until all remaining long blocks are separated by at least three short blocks.

We will use $\ell_B$ to denote the leftmost bit position in block~$B$
and $r_B$ to denote the rightmost bit position in block~$B$.
As long as $B$ is not the leftmost block, we have
$a_{\ell_B} / a_{r_B}\geq n^4$.

\subsection{Definition of $\Phi_f$}
\label{sec:defPhif}

We will define weights 
 $w_1, \ldots, w_n \in \R$ such that $\Phi_f(x) = \sum_{i = 1}^n w_i x_i$. We call the $w_i$  \emph{weights} to distinguish them from the \emph{coefficients} $a_1, \ldots, a_n$ of $f$. 

We define the weights $w_1,\ldots,w_n$ as follows, starting
with $w_1=1$.
Suppose that bit position~$i$ is in block~$B$, that $i\neq r_B$, and that $w_{r_B}$ is already defined.
If block~$B$ is a long block, or is immediately to the left of a long block,
then we define $w_i$  by $w_i = w_{r_B} a_i / a_{r_B}$.
We call this the \emph{copy regime} since $w_i/w_{r_B} = a_i/a_{r_B}$.
Otherwise, we are in the \emph{damped regime} and we define $w_i$ by
$$w_i = w_{r_B} \min\{K^{(i-r_B)c/n},a_i/a_{r_B}\},$$
where~$K$ is the parameter from Section~\ref{sec:param}. 

It will be a major effort in the remainder of the paper to show that this $\{\Phi_f \mid f\in F\}$ is a feasible 
family of drift functions for the EA. It is easier to see
that $\{\Phi_f\}$ is piece-wise polynomial with respect to the EA, so we do this next.

\begin{lemma}
Let $F$ be a linear family of objective functions over
bit~strings. Consider the (1+1) EA for minimising~$F$ with independent bit-mutation
rate $p_n=c/n$. The family $\Phi = \{\Phi_f\}$ of drift functions
constructed above is piece-wise polynomial with respect to the EA.
\label{lem:gotpiecewise}
\end{lemma}

\begin{proof}

Let $k = 6\lceil 1/\gamma \rceil + 1$.
We now construct a family of fitness-based $k$-partitions for~$F$.

Let $f$ be an objective function in~$F$ and let $n=n(f)$.
We now define the partition $\calM^f$.
We call a bit position $i \in [2..n]$ a \emph{jump} (for the objective function~$f$) if 
\begin{itemize}
\item $i$ is in a copy regime, and
\item $w_i/w_{i-1}>n^2$.
\end{itemize}
By the construction of the blocks, bit position~$i$ is
the leftmost bit position of a  miniblock contained in 
either (1)
a long block, or (2) a short block immediately to the left of a long block. 
Since there are at most $\lceil1/\gamma\rceil$ long blocks, there are at most $k-1$ jumps. 
(The easiest way to see this is to think about the original long blocks, prior to any merges.
Each block contains two miniblocks. Within a long block~$B$, there may be two jumps, and there
may be two in each of the \emph{two} blocks to the left of~$B$
--- the block immediately to the left of~$B$ is always in the copy regime, but the block 
to its left may also be merged into a long block with~$B$.)
Suppose there are $k'$ jumps, and let $\calM^f_j=\emptyset$, for $k'+1<j\leq k$.

Let $i_1, \ldots, i_{k'}$ be an increasing enumeration of the jumps. 
Set $i_0 = 1$ and $i_{k'+1} = n+1$ to ease the following definition. 
For $j = 1, \ldots, k'+1$, let 
$N_j$ be $\{i_{j-1},\ldots,i_j-1\}$ and define
$$\calM_j^f = \{x \in \{0,1\}^n \mid \exists i \in 
N_j: x_i = 1 \wedge \forall i \ge i_j : x_i = 0\}.$$ 
Let $\calM^f_0 = \{{\bf 0}\}$. 
Informally, $N_j$ is the set of bit positions starting at the jump $i_{j-1}$ and going up
to, but not including, the jump $i_j$. So
$\{N_j \mid 1 \leq j \leq k'+1\}$ is a partition of the bit positions.
Then $\calM_j^f$ is the set of bit~strings~$x$ which have the leftmost ``$1$''-bit in 
$N_j$.

In order to show that $\calM = \{\calM^f \mid f\in F\}$ is a family of fitness-based $k$-partitions
for~$F$, we need only  show that the following condition is satisfied:
for all $i<j$, $x \in M^f_i$ and $y \in M^f_j$, 
we have $f(x) < f(y)$. 
The condition follows from the fact that 
$a_i/a_{i-1} = w_i/w_{i-1} > n^2$ for
all jumps $i$.
 
In order to show that $\Phi$ is piece-wise polynomial
with respect to the EA, 
it remains to prove that, for every  $j\in\{1,\ldots,k\}$,
$\ln(\max \Phi_f(M^f_j)) - \ln(\min \Phi_f(M^f_j)) = O(\log n(f))$.
Fix any such~$j$.
Let $r_f =   \max \Phi_f(M^f_j) / \min \Phi_f(M^f_j)$. We show that $r_f$ is upper-bounded
by a polynomial in~$n$.

For a set of bit positions $I\subseteq \{1,\ldots,n\}$, let
$\min I$ denote the minimum element in~$I$ and let $\max I$ denote the maximum element.
Since $w_1 \leq \ldots \leq w_n$, 
$\min \Phi_f(M^f_j) = w_{\min N_j} = w_{i_{j-1}}$.
Similarly, 
$\max \Phi_f(M^f_j) = \sum_{i = 1}^{\max N_j} w_i \le n w_{\max N_j} = n w_{i_j -1}$.
Hence $r_f \le n w_{\max N_j}/w_{\min N_j}$. 

We rewrite 
\begin{equation}
\label{eq:product}
\frac{w_{\max N_j}}{w_{\min N_j}} = \prod_{B: B \cap N_j \neq \emptyset} 
\frac{w_{\max (B \cap N_j)}}{w_{\min (B\cap N_j)}},
\end{equation}
where $B$ runs over all miniblocks that have a non-empty intersection with $N_j$. Note that the above is true because adjacent miniblocks intersect in exactly one bit position. 

If $B$ is a miniblock in a damped regime, then $w_{\max (B \cap N_j)}/w_{\min (B\cap N_j)} \le w_{\ell_B}/w_{r_B} = K^{(\ell_B - r_B)c/n}$. In consequence, the contribution of all weights in damped regimes to (\ref{eq:product}) is at most a factor $K^{c}$.

What remains is the contribution of miniblocks in long blocks and in those short blocks immediately to the left of a long block. Let $B$ be such a miniblock. 
If  $B\cap N_j = \{\ell_B\}$ then $w_{\max (B \cap N_j)}/w_{\min (B\cap N_j)} = 1$.
Otherwise, note that 
$$\frac{w_{\max (B \cap N_j)}}{w_{\min (B\cap N_j)}} \leq 
\frac{w_{\ell_B}}{w_{r_B}} =
\left(\frac{w_{\ell_B}}{w_{\ell_B-1}} \right)\left(\frac{w_{\ell_B-1}}{w_{r_B}}\right).$$
The first factor is at most $n^2$, since $\ell_B$ is not a jump, the second factor is at most $w_{\ell_B-1}/w_{r_B}=a_{\ell_B-1}/a_{r_B} \le n^2$ by the definition of a miniblock.
\end{proof}

\subsection{Auxiliary results concerning the weights $w_i$}
\label{sec:revlabel}

Fix an objective function~$f\in F$ and let $n=n(f)$.
We will assume that $n$ is sufficiently large with respect to the constants~$c$, $\eps$, $K$ and $\gamma$ since our objective
is to construct a family~$\Phi$ of feasible drift functions for the EA and the definition of such a family (Definition~\ref{def:feasdrift}) is only concerned with sufficiently large~$n$.
The definition of $\Phi_f$   allows us to prove a number of useful facts.  
The first of these uses a geometric series to bound sums of weights in the damped regime.

\begin{lemma}
\label{lwsumgeometric}
Let $B_0,\ldots,B_k$ be a consecutive sequence of blocks (left to right) in the damped regime with 
$\ell_{B_0} = r_{B_k} + t$.
Then \[\sum_{j  \in B_0\cup \ldots \cup B_k} w_j \le K^{tc/n} w_{r_{B_k}} \left(\frac{n}{c \ln K}+1\right).\]
\end{lemma}
\begin{proof}
 
For $0\leq h \leq t$
we have $w_{\ell_{B_0} - h} \leq K^{tc/n} w_{r_{B_k}} K^{-hc/n}$.
Now
\[
\sum_{j  \in B_0\cup \ldots \cup B_k} w_j 
  \leq 
K^{tc/n} w_{r_{B_k}}
\sum_{h=0}^\infty K^{-ch/n}
= K^{tc/n} w_{r_{B_k}}
\frac{1}{1-K^{-c/n}}
.\]
Now $K^{c/n} = 
e^{(\ln K)c/n} \geq 1 + (\ln K)c/n$,
so  
$$
\frac{1}{1-K^{-c/n}} \leq
\frac{1}{1-\frac{1}{1+(\ln K)c/n}}
= 
\left(
\frac{n}{c \ln K}+1
\right).
$$
\end{proof}

The next lemma gives the relationship between the leftmost weight and the rightmost weight
in a block in the damped regime.
\begin{lemma}
\label{lem:whatever}
If $B$ is a block in the damped regime with $\ell_B = r_B + t$
and $B$ is not the leftmost block, then $w_{\ell_B}=  K^{t c /n} w_{r_B}$.
\end{lemma}
\begin{proof}
This follows from the definition of the weights in the damped regime, since
$$\frac{a_{\ell_B}}{a_{r_B}} \geq n^4 \geq K^c \geq K^{tc/n}.$$
The second inequality follows from our assumption (at the beginning of Section~\ref{sec:revlabel}) that $n$ is sufficiently large with respect to~$K$ and~$c$.
\end{proof}

Lemmas~\ref{lwsumgeometric} and~\ref{lem:whatever} give the following corollary.
\begin{corollary}
\label{cwsumgeometric}
Let $B_0,\ldots,B_k$ be a consecutive sequence of blocks (left to right) in the damped regime with 
$\ell_{B_0} = r_{B_k} + t$. If $B_0$ is not the leftmost block
then \[\sum_{j  \in B_0\cup \ldots \cup B_k} w_j \le w_{\ell_{B_0}} \left(\frac{n}{c \ln K}+1\right).\]
\end{corollary}

Corollary~\ref{cwsumgeometric}
gives the following upper bound for the sum of all weights contained in, and to the right of, a short block.

\begin{lemma}\label{lwsumrightend}
Let $B$ be a short block that is not the leftmost block. 
Then \[\sum_{j \le \ell_B} w_j \le w_{\ell_B} \left(\frac{n}{c \ln K} + 1 + \gamma n + n^{-3}\right).\]
\end{lemma}

\begin{proof}
 If there is no long block to the right of $B$, then 
$B$ and all of the blocks to its right are in the damped regime, so
the result follows immediately from 
Corollary~\ref{cwsumgeometric}. Assume therefore that there is a long block to the right of $B$. Let $L$ be the long block which is closest to $B$ on its right. Let $S$ be the short block immediately to the left of $L$. Note that $S$ might be the same block as $B$. 
  
Suppose $j\in L$.
Recall that for all $h, k \in L \cup S$, we have $\frac{w_h}{a_h} = \frac{w_k}{a_k}$. Thus,
since $S$ is not the leftmost block, 
$$w_j = \frac{w_j }{a_j} a_j
 \le \frac{w_j }{a_j} n^{-4} a_{\ell_S} 
 = n^{-4} w_{\ell_S} \le n^{-4} w_{\ell_B}.$$ 
Since the $w_j$'s increase with~$j$,
we conclude that $w_j \leq n^{-4} w_{\ell_B}$ for any $j\leq \ell_L$.
Thus, $\sum_{j \le \ell_L} w_j \le n^{-3} w_{\ell_B}$.
  
 Using the fact that $S$ is short and the monotonicity of $w$, we  deduce $\sum_{j \in S} w_j \le \gamma n w_{\ell_B}$.  Combining this with 
Corollary~\ref{cwsumgeometric},
we obtain \[\sum_{j \le \ell_B} w_j \le w_{\ell_B} \left(\frac{n}{c \ln K} +1 \right)+ \gamma n w_{\ell_B} + w_{\ell_B} n^{-3}.\]
\end{proof}

\section{Feasible Drift}

Our objective in this section is to prove the following lemma, which is the heart of the proof of our main result.

\begin{lemma}
Let $F$ be a linear family of objective functions over
bit~strings. Consider the (1+1) EA for minimising~$F$ with independent bit-mutation
rate $p_n=c/n$. There is a function $\nu(n) = O(n)$ such that
the family $\Phi = \{\Phi_f\}$ of drift functions
constructed above is $\nu$-feasible for the EA.  
\label{lem:gotfeasible}
\end{lemma}

Consider running the EA with input~$f$ with $n=n(f)$. We use the following notation.
The state after $t$ steps is a binary string $x[t] = x_n[t]\ldots x_1[t]$. 
Recall from Section~\ref{sec:param} that we write bit~strings
as words from most significant bit (``leftmost bit'') to least significant.
In the $(t+1)$'st step of the algorithm, 
the bits of a binary string $y[t+1]=y_n[t+1]\ldots y_1[t+1]$ encoding the mutation mask are chosen independently.
The probability that $y_i[t+1]=1$ is $p_n=c/n$.
Then $x'[t+1]$ is formed from $x[t]$ by flipping the bits that are $1$ in
string~$y[t+1]$.
That is,
$x'_n[t+1] \ldots x'_1[t+1] = (x_n[t]\oplus y_n[t+1])\ldots (x_1[t] \oplus y_1[t+1])$.
Let $A_{t+1}$ be the event
that 
$\sum_i a_i x'_i[t+1] \leq \sum_i a_i x_i[t]$.
We say that the mutation  in step~$t+1$ is ``accepted'' in this case.
If $A_{t+1}$ occurs, then $x[t+1]=x'[t+1]$.
Otherwise, $x[t+1]=x[t]$. Of course, the coefficients~$a_i$, and therefore $A_{t+1}$ itself,
depends implicitly on~$f$.
Suppose that $x[t]$ is not the all-zero string.  For a bit position~$i$
with $x_i[t]=1$,
let $I_i[t+1]$ be the event
$$y_i[t+1]=1 \wedge \forall j\in \{i+1,\ldots, n\}: 
(x_j[t]=1) \Rightarrow
(y_j[t+1]=0).$$
 $I_i[t+1]$ is the event that $i$ is the leftmost `1'
to be considered for a flip in step~$t+1$.
Finally, let $I'_\ell[t+1]$ be the event 
$$\forall j\in \{\ell+1,\ldots, n\}:
(x_j[t]=0) \Rightarrow
(y_j[t+1]=0)
.$$
$I'_\ell[t+1]$ is the event that the `0' bits to the left of $\ell$ are not considered for a flip in step~$t+1$.
Note that $\Pr(I'_\ell[t+1])\geq {(1-p_n)}^n$ and that, given~$x[t]$, the event $I'_\ell[t+1]$ is independent of $I_i[t+1]$
for any~$i$
(the event $I_i[t+1]$ constrains $y_j[t+1]$ for some $j$ with $x_j[t]=1$,
whereas the event $I'_\ell[t+1]$ constrains $y_j[t+1]$ for~$j$ with $x_j[t]=0$). 
However, these events are not independent if we condition on $A_{t+1}$,
as the following simple observation shows.
  
\begin{lemma}\label{liprime}
  Let $i$ be a bit position contained in some block $B$. Assume that there is a block $L$ immediately to the left of $B$. Then $I_i[t+1]$ and $A_{t+1}$ implies $I'_{\ell_L}[t+1]$.
\end{lemma}

\begin{proof}
  There is nothing to show if $L$ is the leftmost block. Hence assume that it is not. Then in particular, $a_{\ell_L} \ge n^4 a_{r_L}$. 
  
  Assume that  
  $I_i[t+1]$ occurs and $I'_{\ell_L}[t+1]$ does not. Let $k > \ell_L$ be such that $y_k[t+1] = 1$ and $x_k[t] = 0$. Then $\sum_{j = 1}^n a_j (x_j[t+1] - x_j[t]) \geq  a_k -  \sum_{j \le i} a_j 
  \geq a_k - n a_i > 0$, because $a_k \ge a_{\ell_L} \ge n^4 a_{r_\ell} \ge n^4 a_i$. Hence this mutation is not accepted, that is, $A_{t+1}$ does not occur.
\end{proof}

Recall that $\eps\in(0,1)$, $K$, and $\gamma$ are parameters defined in Section~\ref{sec:param}.
We take $\eps$ to be ``sufficiently small''. Then $K\geq 1$ is taken to be ``sufficiently large''
(depending on~$c$ and~$\eps$) and then $\gamma\in(0,1) $ is taken to be ``sufficiently small''
(depending on~$c$, $\eps$ and $K$).
Finally, we take $n_0>1$ to be any integer which is ``sufficiently large'' with respect to 
all of these parameters.
The actual constraints that we use 
(to determine what is ``sufficiently large'' and what is ``sufficiently small'') will be spelled out below.
Note that $(1-\frac{c}{n})^{n}$ approaches $\exp (-c)$ from below as $n\rightarrow \infty$.
We choose $n_0$ so that $(1-\frac{c}{n_0})^{n_0}$ is ``sufficiently close'' to 
$\exp (-c)$ (with respect to $c$, $\eps$ and $K$).
We can conclude from this that $(1-\frac{c}{n})^{n}$ is ``sufficiently close'' to $\exp(-c)$ for
any $n\geq n_0$. Similarly, 
$(1-\frac{c}{n})^{3n}$ approaches $\exp (-3c)$ from below as $n\rightarrow \infty$.
We will choose $n_0$ to ensure that, for $n\geq n_0$, this
is ``sufficiently close'' to $\exp(-3c)$.

\begin{proof}[Proof of Lemma~\ref{lem:gotfeasible}]
The first two conditions in Definition~\ref{def:feasdrift} follow from the construction of~$\Phi_f$ in Section~\ref{sec:defPhif}.
The third condition follows from Lemma~\ref{lem:drift} below. 
\end{proof}

The following lemma is the main ingredient in the short proof of Lemma~\ref{lem:gotfeasible} above.
It establishes the third condition in Definition~\ref{def:feasdrift}, so it allows us to conclude that $\Phi$ is $\nu$-feasible
for the EA. Since by Lemma~\ref{lem:gotpiecewise}, $\Phi$ is also piece-wise polynomial with respect to the EA,
Lemma~\ref{lem:piecewise} enables us to repeatedly apply Lemma~\ref{lem:drift} to bound
the expected optimisation time of the EA.

\begin{lemma}\label{lem:drift}
Let $F$ be a linear family of objective functions over
bit~strings. Consider the (1+1) EA for minimising~$F$ with independent bit-mutation
rate $p_n=c/n$. 
Let~$f$ be an objective function in~$F$    with $n(f)\geq n_0$.
For all $x \in \{0,1\}^{n(f)} \setminus \{\bf 0\}$,
\[E[\Phi_f(x[t+1]) \mid x[t]=x] \leq 
\left(1-\tfrac1{n(f)} c e^{-3c}{(1-\eps)}^2\right) \Phi_f(x).\]
\end{lemma}
 
\begin{proof} 

Fix $f\in F$ with $n(f)\geq n_0$. Let $n=n(f)$.
Note that, for any fixed $x[t]$,
\begin{equation}\label{eqstart}
E[\Phi_f(x[t]) - \Phi_f(x[t+1])]  =
\sum_{i: x_i[t]=1}
\Pr(I_i[t+1])
E[\Phi_f(x[t]) - \Phi_f(x[t+1])\mid I_i[t+1]],
\end{equation}
since the events $I_i[t+1]$ for 
$1\leq i \leq n$ are disjoint  and 
$\Phi_f(x[t]) = \Phi_f(x[t+1])$ unless one of them occurs.
In each of various cases (see Subsections~\ref{sec:case1} to~\ref{sec:caselast}),
we will show that,  for all $i$ with $x_i[t]=1$,
\begin{equation}
\label{eq:caseeq}
E[\Phi_f(x[t]) - \Phi_f(x[t+1])\mid I_i[t+1]]  \geq (1-p_n)^{2n} w_i (1 - \eps),
\end{equation}
which is greater than or equal to~$0$ since $n\geq n_0>c$ and $\eps<1$.
Using the lower bound
 $\Pr(I_i[t+1]) \geq p_n (1-p_n)^n $, which applies for every~$i$ with $x_i[t]=1$,
 Equations~(\ref{eqstart})
 and~(\ref{eq:caseeq}) give
 
 \begin{align*}
E[\Phi_f(x[t]) - \Phi_f(x[t+1])]  
&\geq  p_n (1-p_n)^n 
\sum_{i:x_i[t]=1}
E[\Phi_f(x[t]) - \Phi_f(x[t+1])\mid I_i[t+1]] \\
&\geq  p_n (1-p_n)^n (1-p_n)^{2n}  (1 - \eps) \Phi_f(x[t]),
\end{align*}
so 
$$E[\Phi_f(x[t+1])] \leq (1 - p_n   (1-p_n)^{3n}  (1 - \eps) 
) \Phi_f(x[t]).$$
Since $(1-p_n)^{3n} \ge e^{-3c}(1-\eps)$ for $n\geq n_0$, 
this will complete the proof.

It remains to prove Equation (\ref{eq:caseeq}). We do this in 
Subsection~\ref{sec:case1} to~\ref{sec:caselast}.
In each case, $B$ is the block containing bit position~$i$, $L$ is the block to
the left of~$B$ (if it exists) and $R$ is the block to the right of~$B$ (if it exists).
Figure~\ref{fig:one} depicts some blocks (two short blocks  
followed by a long block, followed by 
a short block divided into two miniblocks, followed
by another short block).
For each possible location of the bit position~$i$, it names the relevant case.
Every long block is covered by Case~5. 
Blocks to the left of a long block are covered by Case~3 and 
  blocks
immediately to the right 
of a long block
are covered by Case~4, then Case~2. Everything else is covered by Case~1.
 
\begin{figure}

\centering{
 
\begin{tikzpicture} 
\draw(0,0) -- (14,0)--(14,1)--(0,1)--(0,0);
\draw [style=dashed] (11,0) -- (11,1); 
\draw (2,0) -- (2,1);  
\draw (4,0) -- (4,1);  
\draw (10,0) -- (10,1);  
\draw (12,0) -- (12,1);  
\draw (1,0.5) node {\tiny Case $1$};
\draw (3,0.5) node {\tiny Case $3$};
\draw (7,0.5) node {\tiny Case $5$};
\draw (10.5,0.5) node {\tiny Case $4$};
\draw (11.5,0.5) node {\tiny Case $2$};
\draw (13,0.5) node {\tiny Case $1$};
\end{tikzpicture}
}

\caption{The cases that are used to proof Equation (\ref{eq:caseeq}).}
\label{fig:one}
\end{figure}

\end{proof} 
  
For all of the following cases,
fix $f\in F$ with $n(f)\geq n_0$. Let $n=n(f)$. 
Fix $x[t]$ with $x_i[t]=1$ for a bit position~$i$ in block~$B$. 
Recall from the proof of Lemma~\ref{lem:drift} that the goal is to prove~(\ref{eq:caseeq}).
That is, we must show that
$$
 E[\Phi_f(x[t]) - \Phi_f(x[t+1])\mid I_i[t+1]]  \geq (1-p_n)^{2n} w_i (1 - \eps).$$

\subsection{Case 1}
\label{sec:case1}

For this case, assume that $B$ is not long and that blocks adjacent to~$B$ are not long either.

If $B$ is not the leftmost block, then let $L$ be the block to $B$'s left.
The case in which $B$ is the leftmost block is actually easier, but to avoid
repetition, in this case, let $L$ be the \emph{block} consisting of 
the single bit position $\ell_B$. The following argument now applies whether~$L$ is a real block or just
a single bit position.

We will condition on $I_i[t+1]$. By Lemma~\ref{liprime}, we 
know that if this mutation is accepted (so $A_{t+1}$ occurs), then the event
$I'_{\ell_L[t+1]}$ occurs. Also,
$\Pr(I'_{\ell_L}[t+1] \mid I_i[t+1]) \geq (1-p_n)^n$, as we noted earlier.
Thus 
$E[\Phi_f(x[t]) - \Phi_f(x[t+1])\mid I_i[t+1]] $
is equal to
\begin{equation}
   \Pr(I'_{\ell_L}[t+1] \mid I_i[t+1]) \cdot E[\Phi_f(x[t]) - \Phi_f(x[t+1])\mid I_i[t+1], I'_{\ell_L}[t+1]].
\label{7July11}
\end{equation}
Let $P = \Pr( A_{t+1}
\mid I_i[t+1],I'_{\ell_L}[t+1])$. Note that $P \geq {(1-p_n)}^n$
(since, for example, $A_{t+1}$ occurs  
if $y_j[t+1] = 0$ for $j\neq i$).
Now 
$\Phi_f(x[t]) - \Phi_f(x[t+1]) = \sum_{j=1}^n w_j(x_j[t]-x_j[t+1])$.
If $I_i[t+1]$ and $I'_{\ell_L}[t+1]$ occur, then this is 
$\sum_{j\leq \ell_L} w_j(x_j[t]-x_j[t+1])$.
If $A_{t+1}$ also occurs, then $x_i[t]-x_i[t+1]=1$ so this is
$w_i + \sum_{j\leq \ell_L,j\neq i} w_j(x_j[t]-x_j[t+1])$.
Thus,  the quantity in (\ref{7July11}) is at least
\begin{align*}
& (1-p_n)^n \left(w_i P - \sum_{j\leq \ell_L,j\neq i} w_j \Pr(y_j[t+1]=1 \mid I_i[t+1], I'_{\ell_L}[t+1])\right)\\
& \geq
(1-p_n)^n \left(w_i (1-p_n)^n - \sum_{j\leq \ell_L} w_j p_n\right).
\end{align*}

Now, by Lemma~\ref{lwsumrightend}, we have
\[\sum_{j \le \ell_L} w_j \le K^{2 c \gamma} w_{r_B} \left(\frac{2 n}{c \ln K} + 2 + \gamma n + n^{-3}\right).\]
To see this, apply the lemma directly to $L$ if it is not the leftmost block
(and note that $w_{\ell_L} \leq K^{2 c \gamma} w_{r_B} $). 
If $L$ is the leftmost block (and $B$ is not)
then
apply Lemma~\ref{lwsumrightend} to block~$B$ 
(noting that $w_{\ell_B} \leq K^{c \gamma} w_{r_B}$)
and use Lemma~\ref{lwsumgeometric} to sum the weights in~$L$.
Finally, if $B$ is the leftmost block then 
apply Lemma~\ref{lwsumrightend} to the short block
to the right of~$B$ and use Lemma~\ref{lwsumgeometric} to sum the weights in~$B$.

Using this and $w_{r_B}\leq w_i$  we have 
\begin{align*}
& E[\Phi_f(x[t]) - \Phi_f(x[t+1])\mid I_i[t+1]] \\
& \geq
(1-p_n)^n w_i 
\left(
(1-p_n)^n 
-  \frac{2 K^{2 c \gamma} }{ \ln K} 
- 2 \frac  c n K^{2 c \gamma}
-  \gamma c K^{2 c \gamma}
- \frac{c}{n^4}K^{2 c \gamma}
\right).
\end{align*}

By the choice of the parameters in Section~\ref{sec:param},
and since $n\geq n_0$,
each of 
$\frac{2 K^{2 c \gamma} }{ \ln K} $,
$2 \frac  c n K^{2 c \gamma}$,
$\gamma c K^{2 c \gamma}$
and $\frac{c}{n^4}K^{2 c \gamma}$
is at most 
$ 
(1-p_n)^n \eps/4$, so Equation~(\ref{eq:caseeq}) holds, as required.
 To see this, recall (from the text just after Lemma~\ref{liprime} )
that   $\eps$ is taken to be ``sufficiently small'', then $K\geq 1$ is taken to be ``sufficiently large''
(depending on~$c$ and~$\eps$) and then $\gamma\in(0,1) $ is taken to be ``sufficiently small''
(depending on~$c$, $\eps$ and $K$).
Finally, we take $n_0>1$ to be any integer which is ``sufficiently large'' with respect to 
all of these parameters, in particular, guaranteeing that  $(1-p_n)^{n}$ is ``sufficiently close'' to $\exp(-c)$ for
any $n\geq n_0$.  
It is easy to see that $\frac{c}{n^4}K^{2 c \gamma}$
and  $2 \frac  c n K^{2 c \gamma}$ are sufficiently small, since $n_0$ is chosen after the other parameters
(so these terms can be made arbitrarily small as compared to $\exp(-c)\eps/4$).
Similarly, $\gamma c K^{2 c \gamma}$ is sufficiently small because $\gamma$ is chosen  
to be sufficiently small with respect to $\eps$, $c$ and $K$.
Finally, $\frac{2 K^{2 c \gamma} }{ \ln K} $
is sufficiently small because $\gamma$ can be chosen   as small as we like with respect to
the other parameters. 
(That is, \emph{first} $K$ is made sufficiently large with respect to~$c$ and $\eps$ and \emph{then} $\gamma$ is
defined.) 
For example, setting $\gamma = \ln(\tfrac\eps{16}e^{-c} \ln K)/(2 c \ln K)$  gives 
$\frac{2 K^{2 c \gamma} }{ \ln K} = e^{-c} \eps /8$.

\subsection{Case 2}
For this case, assume that the block~$L$, immediately to the left of~$B$, is long,
and that $i$ is in the rightmost miniblock of block~$B$ (which is therefore short).
  
This is very similar to Case~1. As in Case~1, we will condition 
on~$I_i[t+1]$. Where Case~1 uses Lemma~\ref{liprime}, we use exactly the same argument to
show that, if this mutation is accepted (so $A_{t+1}$ occurs), then event
$I'_{\ell_B[t+1]}$ occurs. 
From that point the argument proceeds exactly as in Case~1, replacing ``$\ell_L$'' with
``$\ell_B$''. We use Lemma~\ref{lwsumrightend}
to obtain the upper bound
\begin{align*}
\sum_{j \le \ell_B} w_j 
&\le   w_{\ell_B} \left(\frac{ n}{c \ln K} + 1 + \gamma n + n^{-3}\right)\\
&\le   K^{c \gamma} w_{r_B} \left(\frac{ n}{c \ln K} + 1 + \gamma n + n^{-3}\right).
\end{align*}
The rest of the argument is exactly the same as in Case~1.

\subsection{Case 3}

For this case, assume that $B$ is immediately to the left of a long block~$R$. 
Hence both $B$ and $R$ are in the copy regime.

If $B$ is not the leftmost block, then there is a block $L$ immediately to the left of $B$. Block~$L$ is short, since any pair of long blocks has at three short blocks between.
Thus, $L$ is in the damped regime. If $B$ is the leftmost block, to keep notation simple, we add an artificial block $L = \{\ell_B\} = \{n\}$.

 Note that
\begin{align}
\label{eqCase3}
\sum_{j< r_{R}} w_j
\leq n w_{\ell_{R}} \frac{w_{r_R}}{w_{\ell_R}}
= n w_{\ell_{R}} \frac{a_{r_R}}{a_{\ell_R}}
\leq n^{-3} w_i.
\end{align}
   
Let $Y$ be the set of $n$-bit binary strings so that, if $y[t+1]=y$, then
$I_i[t+1]$ occurs and
$A_{t+1}$ occurs (the move in step~$t+1$ is accepted).
We first analyse the effect of such a mutation. Let $y \in Y$.
As in Case~1, $A_{t+1}$ implies $I'_{\ell_L}[t+1]$. Consequently, we have
   $y_j=0$ for all $j$ that fullfill $j>\ell_L$ or both
   $j>i$ and $x_j[t]=1$.
Thus, by the definition of $A_{t+1}$, we have
\begin{equation*}
\sum_{j\leq \ell_L} a_j ((x_j[t]\oplus y_j) - x_j[t]) \leq 0.
\end{equation*}
We compute
 $$
\sum_{j\in L: y_j=1,x_j[t]=0} a_j +
\sum_{j\in B\cup R: y_j=1,x_j[t]=0} a_j -
\sum_{j \in B \cup R: y_j=1,x_j[t]=1} a_j \leq 
\sum_{j<r_R} a_j
 \leq n a_{r_R}
\leq n^{-3} a_i.$$
Dividing through by $a_i$, we have
$$
\sum_{j\in L: y_j=1,x_j[t]=0} \frac{a_j}{a_i} +
\sum_{j\in B\cup R: y_j=1,x_j[t]=0} \frac{a_j}{a_i} -
\sum_{j \in B \cup R: y_j=1,x_j[t]=1} \frac{a_j}{a_i} \leq 
n^{-3} .$$ Now for $j$ in the copy regime (blocks~$B$ and~$R$), $
a_j/a_i = w_j/w_i$.
Also, for $j\in L$ (which is in the damped regime), 
$$w_j = w_{r_L} \min(K^{(i-r_L)c/n},a_j/a_{r_L})
\leq w_{r_L} \frac{a_j}{a_{r_L}} = w_{r_L} \frac{a_j}{a_i} \frac{w_i}{w_{r_L}} = w_i \frac{a_j}{a_i},$$
so
$a_j/a_i\geq w_j/w_i$.
Hence, replacing $a_j/a_i$ with $w_j/w_i$ and multiplying through by $w_i$, we have
\begin{equation}
\label{blahblah}
\sum_{j\in L: y_j=1,x_j[t]=0} w_j +
\sum_{j\in B\cup R: y_j=1,x_j[t]=0} w_j -
\sum_{j \in B \cup R: y_j=1,x_j[t]=1} w_j  
\leq n^{-3} w_i.
\end{equation}

For the mutation being random (but conditioning on $I_i[t+1]$ and $I'_{\ell_L}[t+1]$), we compute the following.
Let
$E_1 = E\left[\sum_{j \in L \cup B \cup R} w_j (x_j[t]-x_j[t+1])\mid I_i[t+1], I'_{\ell_L}[t+1]\right]$.
Then
\begin{align*}
& E[\Phi_f(x[t])  - \Phi_f(x[t+1])\mid I_i[t+1], I'_{\ell_L}[t+1]] \\
& = E_1 + 
\sum_{j <r_R} w_j
E[x_j[t]-x_j[j+1]
\mid I_i[t+1], I'_{\ell_L}[t+1]
],
\end{align*}
which is at least
$ E_1
 - n^{-3} w_i$ by (\ref{eqCase3}).
Also, 
$$
   E_1=
\sum_{y \in Y} \Pr(y[t+1]=y \mid I_i[t+1], I'_{\ell_L}[t+1])
\sum_{j \in L \cup B \cup R} w_j 
(x_j[t]-(x_j[t]\oplus y_j)).
 $$
 
Each $y$ with $y_j=0$ for all $j\neq i$
contributes at least $(1-p_n)^n w_i$ to the outer sum.
All other strings $y$ contribute at least $-n^{-3} w_i$ by~(\ref{blahblah}).
Now, putting it together, we 
find that 
\begin{align*}
E[\Phi_f(x[t]) &- \Phi_f(x[t+1])\mid I_i[t+1]]  \\
&= \Pr(I'_{\ell_L}[t+1] \mid I_i[t+1]) E[\Phi_f(x[t]) - \Phi_f(x[t+1])\mid I_i[t+1], I'_{\ell_L}[t+1]]\\
&\geq \Pr(I'_{\ell_L}[t+1] \mid I_i[t+1]) (E_1 - n^{-3} w_i)\\
&\geq \Pr(I'_{\ell_L}[t+1] \mid I_i[t+1]) ( (1-p_n)^n w_i-n^{-3} w_i - n^{-3} w_i)\\
&\geq (1-p_n)^n  ( (1-p_n)^n w_i-n^{-3} w_i - n^{-3} w_i).\\
\end{align*}

Now,
$2n^{-3} \leq \eps (1-p_n)^n$,
so  we have established Equation~(\ref{eq:caseeq}), as required.

\subsection{Case 4}
 
For this case, assume that 
the block~$L$, immediately to the left of~$B$, is long,
and that $i$ is in the leftmost miniblock of block~$B$ (which is short).
 
Let $Y$ be the set of $n$-bit binary strings so that, if $y[t+1]=y$, then
$I_i[t+1]$ occurs and
$A_{t+1}$ occurs (the move in step~$t+1$ is accepted).
As in Case~3, $A_{t+1}$ implies $I'_{\ell_L}[t+1]$. Hence 
for every $y\in Y$ we have
   $y_j=0$ if $j>\ell_L$ or if
   $j>i$ and $x_j[t]=1$.
Thus, if $y\in Y$, then, by the definition of $A_{t+1}$, we have
\begin{equation}
\label{myeq1}
0 \leq \sum_{j\leq \ell_L} a_j (x_j[t] - (x_j[t]\oplus y_j)).
\end{equation} 
To derive an upper bound in the right-hand side of Equation~(\ref{myeq1}) we
split the summation into three easily-bounded parts.
The summation over $j\in L-\{r_L\}$ is equal to $ - \sum_{r_L<j\leq \ell_L: y_j=1} a_j $,
the summation over $j\in B$ is at most  $
\sum_{j\in B: x_j[t]=1,y_j=1} a_j$, and
the summation over $j<r_B$ is at most  $ n a_{r_B} \leq a_i/n$.
From~(\ref{myeq1}) we thus have
\begin{equation}
\sum_{j\in L - \{r_L\}: y_j=1} a_j \leq
  \sum_{j\in B: x_j[t]=1,y_j=1} a_j
 + 
  a_i/n.
\label{myeq2}
\end{equation}
 
Define
$$\Psi(y) = 
- (1 + \tfrac 1n - K^{-\gamma c}) K^{\gamma c}w_i - K^{\gamma c} \sum_{j \le \ell_B, j \neq i, y_j=1} w_j.$$
We will show that, for $y\in Y$,
$$\sum_{j \in [n]} 
w_j (x_j[t] - (x_j[t]\oplus y_j)) \geq \Psi(y).$$
 Start by breaking up the left-hand side as
\begin{equation} 
-\sum_{j \in L-\{r_L\}: y_j=1}\!w_j  
+ \sum_{j\in B: x_j[t]=1,y_j=1}\!w_j
- \sum_{j\in B: x_j[t]=0,y_j=1}\!w_j
+ \sum_{j<r_B} w_j (x_j[t] - (x_j[t]\oplus y_j)). 
\label{myeq3} 
\end{equation}

Recall that for $j \in L$, we have $w_j = \frac{w_{r_L}}{a_{r_L}} a_j$, whereas for $j \in B$, we have 
$$a_j \le a_{r_L} = \frac{a_{r_L}}{w_{r_L}} w_{r_L} \leq  
\frac{a_{r_L}}{w_{r_L}} w_{r_B} K^{(r_L - r_B)c/n} \le  
\frac{a_{r_L}}{w_{r_L}} w_j K^{\gamma c},$$
where the final inequality uses the fact that $B$ is short, that is, $r_L - r_B \leq \gamma n$.  

Thus, the sum of the first two terms in~(\ref{myeq3}) is at least
$$  
-\frac{w_{r_L}}{a_{r_L}} 
\sum_{j \in L-\{r_L\}: y_j=1} a_j  
+ 
\frac{w_{r_L}}{a_{r_L}} K^{-\gamma c}
\sum_{j\in B: x_j[t]=1,y_j=1} a_j,$$
and by~(\ref{myeq2}), this is at least
$$ \left(
-\frac{w_{r_L}}{a_{r_L}} (1-K^{-\gamma c})
\sum_{j\in B: x_j[t]=1,y_j=1} a_j \right)
- 
\frac{w_{r_L}}{a_{r_L}}
\frac{a_i}{n},$$
which 
is at least 
$$ 
- 
\frac{w_{r_L}}{a_{r_L}}
\frac{a_i}{n}
-\frac{w_{r_L}}{a_{r_L}} (1-K^{-\gamma c})
a_i
-\frac{w_{r_L}}{a_{r_L}} 
\sum_{j\in B-\{i\}: x_j[t]=1,y_j=1} a_j.
$$
Upper-bounding $a_j$ with 
$\frac{a_{r_L}}{w_{r_L}} w_j K^{\gamma c} $ in the last term, 
we find that~(\ref{myeq3}) is at least  
\begin{align*}
- \frac{w_{r_L}}{a_{r_L}}\frac{a_i}{n}
&- \frac{w_{r_L}}{a_{r_L}} (1-K^{-\gamma c})a_i
- K^{\gamma c}\sum_{j\in B-\{i\}: x_j[t]=1,y_j=1} w_j \\
&- \sum_{j\in B: x_j[t]=0,y_j=1} w_j
+ \sum_{j<r_B} w_j (x_j[t] - (x_j[t]\oplus y_j)). \end{align*}
Combining the summations,
this is
at least
$$-\frac{w_{r_L}}{a_{r_L}}\tfrac 1n a_i - (1- K^{-\gamma c}) \frac{w_{r_L}}{a_{r_L}} a_i - K^{\gamma c} \sum_{j \le \ell_B, j \neq i, y_j=1} w_j.$$
Upper-bounding $a_i$ with $a_{r_L}$, the
first two terms are at least
$ - (1 + \tfrac 1n - K^{-\gamma c}) w_{r_L} $.
Then upper-bounding $w_{r_L}$ with $K^{\gamma c} w_i$,
the whole thing is at least $\Psi(y)$.

We have shown that, for $y\in Y$,
$$\sum_{j \in [n]} 
w_j (x_j[t] - (x_j[t]\oplus y_j)) \geq \Psi(y).$$
Suppose that $y$ is an $n$-bit binary string such that, if $y[t+1]=y$, then 
$A_{t+1}$ does not occur.
In this case, we also have
$$\sum_{j \in [n]} 
w_j (x_j[t] - (x_j[t]\oplus y_j)) = 0 \geq \Psi(y),$$
since $\Psi(y)<0$.

Now let $y[t+1]$ be random as constructed by the algorithm. Denote by $y^*$ the bit~string that contains exactly one one-entry, namely the one on position $i$. Let $P$ be the probability (conditional on $I_i[t+1]$) that $y[t+1] = y^*$. Now 
\begin{align*}
  E[\Phi_f&(x[t]) - \Phi_f(x[t+1])\mid I_i[t+1]] \\
               &= \sum_{y\in Y}   \Pr(y[t+1]=y \mid    I_i[t+1]) \sum_{j \in [n]} w_j (x_j[t] - (x_j[t]\oplus y_j))\\
&= \sum_{y\in Y} \Pr(y[t+1]=y \mid    I_i[t+1]) \Psi(y) +\\
&  \quad\quad\quad   \sum_{y\in Y} \Pr(y[t+1]=y \mid    I_i[t+1])\left(\sum_{j \in [n]} 
w_j (x_j[t] - (x_j[t]\oplus y_j))- \Psi(y) \right)\\
&=  E[\Psi(y[t+1]) \mid I_i[t+1]] +
\sum_{y\in Y} \Pr(y[t+1]=y \mid    I_i[t+1])
\left(\sum_{j \in [n]} 
w_j (x_j[t] - (x_j[t]\oplus y_j))
- \Psi(y) \right)
\\
   &\geq E[\Psi(y[t+1]) \mid I_i[t+1]] + P(-\Psi(y^*)+\Phi_f(x[t])-\Phi_f(x[t]\oplus y^*)) \\
  &\geq - (1 + \tfrac 1n - K^{-\gamma c}) K^{\gamma c}w_i - K^{\gamma c} \sum_{j \le \ell_B, j \neq i} \tfrac{c}{n} w_j + P((1 + \tfrac 1n - K^{-\gamma c}) K^{\gamma c}w_i + w_i),
\end{align*}
where the first inequality comes by ignoring terms $y \in Y-\{y^*\}$ (since these are non-negative).

Consider the first term,
$$- (1 + \tfrac 1n - K^{-\gamma c}) K^{\gamma c}w_i
=
- (K^{\gamma c} -1 + \tfrac {K^{\gamma c}}n  )  w_i.$$
By our choice of $\gamma$, $K^{\gamma c}-1$ is very small (see the discussion at the end of Case~1). 
Since $n\geq n_0$, 
$$K^{\gamma c} -1 + \tfrac {K^{\gamma c}}n   \leq (\eps/3) (1-\tfrac{c}{n})^n.$$
Now by the definitions of $I_i[t+1]$ and $y^*$,
$P =(1-p_n)^{n-\zeta-1}$, where $\zeta$ is the number of bits $j>i$ such that $x_j[t]=1$.
Thus, $P\geq {(1-p_n)}^{n}$ so 
$$(\eps/3) (1-\tfrac{c}{n})^n
\leq (\eps/3) P.$$ 
We conclude that
the first term 
 is at least  
$-(\eps/3) P w_i$.
Using Lemma~\ref{lwsumrightend} and $w_{\ell_B} \le K^{\gamma c} w_i$, we obtain  
$$K^{\gamma c} \tfrac{c}{n}\sum_{j \le \ell_B} w_j \le w_i 
\left(
\frac{K^{2\gamma c}}{ \ln K} + 
\frac{c K^{2\gamma c}}{n} + K^{2\gamma c} \gamma c  
+ \frac{c K^{2\gamma c}}{n^{4}}\right).$$
Given the constraints on our parameters (see the discussion at the end of Case~1), each of the  
four summands,  
$\frac{K^{2\gamma c}}{ \ln K}$,
$\frac{c K^{2\gamma c}}{n}$,  $K^{2\gamma c} \gamma c  $ and
$\frac{c K^{2\gamma c}}{n^{4}}$, is at most $(\eps/12) P  $. Thus, the
second term, 
$  - K^{\gamma c} \sum_{j \le \ell_B, j \neq i} \tfrac{c}{n} w_j$,
is also at least $-(\eps/3) P w_i$.
In a similar way, we see that the third term, 
$$P((1 + \tfrac 1n - K^{-\gamma c}) K^{\gamma c}w_i + w_i),$$
is at least $P w_i(1-\eps/3)$. We conclude that 
$$E[\Phi_f(x[t]) - \Phi_f(x[t+1]) \mid I_i[t+1]] \ge P w_i (1 - \eps),$$  
which establishes Equation~(\ref{eq:caseeq}), as required.

\subsection{Case 5}
\label{sec:caselast}

For this case, assume that $B$ is a long block.

To the right of $B$, there might be a short block $R$, otherwise 
$r_B=1$ and
we define $R  = \{r_B\}$ to ease notation. To the left of $B$, there might be a short block $L$, otherwise 
$\ell_B=n$ and
we define $L = \{\ell_B\}$ to ease notation.

Let $Y$ be the set of $n$-bit binary strings so that, if $y[t+1]=y$, then
$I_i[t+1]$ occurs and
$A_{t+1}$ occurs (so the move in step~$t+1$ is accepted).
As in Case~4, $A_{t+1}$ implies $I'_{\ell_L}[t+1]$. Hence 
for every $y\in Y$ we have
   $y_j=0$ for $j>\ell_L$ and 
   for all $j>i$ satisfying $x_j[t]=1$.
Thus, if $y\in Y$, then, by the definition of $A_{t+1}$, we have
\begin{align}
\label{myeq1case5}
0 &\leq \sum_{j\leq \ell_L} a_j (x_j[t] - (x_j[t]\oplus y_j)) \nonumber\\
&\leq \sum_{r_R \le j\leq \ell_L} a_j (x_j[t] - (x_j[t]\oplus y_j)) + a_i n^{-3}\nonumber\\
&\leq \sum_{r_B \le j\leq \ell_L} a_j (x_j[t] - (x_j[t]\oplus y_j)) + \sum_{j \in R; y_j = 1; x_j[t]=1} a_j + a_i n^{-3}.
\end{align}

We will use the fact that for $j \in L \cup B$, we 
have $w_j = \frac{w_{r_B}}{a_{r_B}} a_j$ since we are in the copy regime, whereas for $j \in R$, we 
are in the damped regime, so we
have 
$$a_j \le a_{r_B} = \frac{a_{r_B}}{w_{r_B}} w_{r_B} \leq  
\frac{a_{r_B}}{w_{r_B}} w_{r_R} K^{(r_B - r_R)c/n} \le  
\frac{a_{r_B}}{w_{r_B}} w_j K^{\gamma c}.$$  

Plugging this into (\ref{myeq1case5}), we obtain 
\begin{align*}
\sum_{r_B \le j\leq \ell_L}& w_j (x_j[t] - (x_j[t]\oplus y_j)) = 
\frac{w_{r_B}}{a_{r_B}} \sum_{r_B \le j\leq \ell_L} a_j (x_j[t] - (x_j[t]\oplus y_j)) \\ 
\ge &- \frac{w_{r_B}}{a_{r_B}}\left(\sum_{j \in R; y_j = 1; x_j[t]=1} a_j + a_i n^{-3}\right) 
\ge - K^{\gamma c} \sum_{j \in R; y_j = 1; x_j[t]=1} w_j - w_i n^{-3}.
\end{align*}

Let $\Psi(y) = -K^{\gamma c} \sum\limits_{j \le \ell_R; y_j = 1} w_j - w_i n^{-3}$.
From the above,
\begin{align*}
\sum_{j\leq \ell_L} &w_j (x_j[t] - (x_j[t]\oplus y_j)) \\
&\ge (1-K^{\gamma c}) \sum_{j \in R; y_j = 1; x_j[t]=1} w_j - \sum_{j \in R; y_j = 1; x_j[t]=0} w_j - w_i n^{-3} 
- \sum_{j < r_R; y_j = 1} w_j\\
&\ge   \Psi(y).
\end{align*}

We have shown that, if $y\in Y$ (so $I'_{\ell_L}[t+1]$ occurs), then
$$\sum_{j\in[n]} w_j (x_j[t] - (x_j[t]\oplus y_j)) \geq \Psi(y).$$

Suppose now that $y$ is an $n$-bit binary string such that, if $y[t+1]=y$, then 
$A_{t+1}$ does not occur.
In this case, we also have
$$\sum_{j \in [n]} 
w_j (x_j[t] - (x_j[t]\oplus y_j)) = 0 \geq \Psi(y),$$
since $\Psi(y)<0$.
 
Now let $y[t+1]$ be random as constructed by the algorithm. Denote by $y^*$ the bit~string that contains exactly one ``1''-entry, namely on position $i$. Let $P$ be the probability (conditional on $I_i[t+1]$) that $y[t+1] = y^*$. Now, as in Case~4,
\begin{align*}
  E[\Phi&(x[t]) - \Phi_f(x[t+1]) \mid I_i[t+1]] \\
  &\ge E[\Psi(y[t+1]) \mid I_i[t+1]] + P(-\Psi(y^*)+\Phi_f(x[t])-\Phi_f(x[t]\oplus y^*)) \\
  &\geq  - K^{\gamma c} \sum_{j \le \ell_R} \tfrac{c}{n} w_j - n^{-3}w_i + P w_i. 
\end{align*}

Using Lemma~\ref{lwsumrightend} and $w_{\ell_R} \le w_i$, we obtain  
$$K^{\gamma c} \frac{c}{n}\sum_{j \le \ell_R} w_j 
\le w_i \left(
\frac{K^{\gamma c}}{ \ln K} + 
\frac{c K^{\gamma c} }{n} + K^{\gamma c}\gamma c  
+ \frac{c K^{\gamma c}}{n^2}\right).$$ 
Since each of the summands, 
$\frac{K^{\gamma c}}{ \ln K}$, $
\frac{c K^{\gamma c} }{n}$,  $K^{\gamma c}\gamma c  $, 
$\frac{c K^{\gamma c}}{n^2}$ 
and $n^{-3}$
is at most $(\eps/5)P$ (for $n\geq n_0$),
we have  $E[\Phi_f(x[t]) - \Phi_f(x[t+1]) \mid I_i[t+1]] \ge P w_i (1 - \eps)$, which gives 
Equation~(\ref{eq:caseeq}), as required.

The cases that we have just completed conclude the proof of Lemma~\ref{lem:drift},
which was used in the proof of Lemma~\ref{lem:gotfeasible}.
We are now ready to prove Theorem~\ref{thm:main}.
\begin{proof}[Proof of Theorem~\ref{thm:main}]
By Lemma~\ref{lem:gotfeasible}, there is a function $\nu(n) = O(n)$ such that
the family $\Phi = \{\Phi_f\}$ of drift functions
that we have constructed   is $\nu$-feasible for the EA.  
By Lemma~\ref{lem:gotpiecewise} this family of drift
functions is piece-wise polynomial with respect to the EA.
The result now follows from Lemma~\ref{lem:piecewise} (using Definition~\ref{usedef}).
\end{proof}
  
\section{A Simple Lower Bound}\label{sec:lower}
 
The following theorem complements Theorem~\ref{thm:main}, showing that it cannot be improved by more than a constant factor. This extends Lemma 10 in~\cite{DJW02} using the same proof idea.

\begin{theorem}
\label{thm:lower}
Let~$c$ be a positive constant.
Let $\tilde c = \max\{1,c\}$. 
Let~$F$ be a family of linear objective functions over bit strings.
Consider the (1+1) EA for minimising $F$ with independent bit-mutation rate
$p_n = c/n$. 
There is a constant~$n_0$ such that, for any $f\in F$ with $n(f)\geq n_0$,
 the probability that  the optimisation time
is at most  $ n(f) \ln(n(f)) / 
( {2 (\tilde c + 1)}  )$  is at most $\exp(-{n(f)}^{\Omega(1)})$.
\end{theorem}

\begin{proof}

Let $n_0$ be any integer so that
$(1-\tfrac{\tilde c}{n_0})^{n_0} \geq \exp(-(\tilde c+1))$.
It is easy to see that such an~$n_0$ exists,
since $(1-\tfrac{\tilde c}{n_0})^{n_0}$
converges, from below, to $\exp(-\tilde c)$, as $n\rightarrow \infty$.
Consider an input $f\in F$ 
with $n(f)\geq n_0$. Let $n=n(f)$ and let
$$T = \frac{1}{2 (\tilde c + 1)} n \ln n.$$
 
The probability that a particular bit position 
is not touched by any mutation step
during $T$ iterations is at least 
$$(1-p_n)^T \geq (1-\tilde c/n)^T \ge \exp(-(\tilde c+1) T/n)
= n^{-1/2}.$$

By a Chernoff bound, the probability that the
initial solution~$x$ (which is chosen uniformly at random from $\{0,1\}^n$) has
at least $n/3$ bit positions that are one is at least
$1 - \exp(- n/36)$.  
The probability that all of these bits 
are touched in $T$ mutation steps is at most 
$(1 - n^{-1/2})^{n/3} \le \exp(-(1/3) n^{1/2})$. 

Thus, the probability that the optimum is found in~$T$ steps is at most
$\exp(- n/36) + \exp(-(1/3) n^{1/2})$. \end{proof}

 \section{Conclusion}

Let~$c$ be a positive constant.
Let~$F$ be a family of linear objective functions over bit strings.
Theorem~\ref{thm:main} shows that the (1+1) EA for minimising $F$ with independent bit-mutation rate
$p_n = c/n$ 
has expected optimisation time $O(n(f) \log n(f))$.
The proof of the theorem constructs a feasible family of drift functions for the EA that is piece-wise polynomial.
The construction of the drift functions depends on the relevant objective functions.
By reproving a classical drift theorem, we  
also show that our bound on the expected optimisation time
 also holds with high probability. 
 This version of the drift theorem makes it easy to  extend a number of other classical bounds stemming from drift or ``expected multiplicative weight decrease'' 
arguments to also hold with high probability, instead of only with expectation (see~\cite{PPSNtail}). We expect this version of the drift theorem to become a useful tool in the theory of evolutionary algorithms.
 
\subsection*{Acknowledgements}

The authors would like to thank Daniel Johannsen for several useful comments.

\bibliographystyle{abbrv}
\bibliography{thispaper}
 
\end{document}